\documentclass[journal,onecolumn]{IEEEtran}
\usepackage{graphicx}
\usepackage{subcaption}
\usepackage[dvipsnames]{xcolor}
\usepackage{latexsym,amssymb}
\usepackage{amsmath}
\usepackage{amsthm,amssymb}
\usepackage{rotating}
\usepackage{caption}
\usepackage{hyperref}
\usepackage{calrsfs}

\newtheorem{theorem}{Theorem}

\newtheorem{remark}{Remark}
\DeclareGraphicsExtensions{.pdf,.png,.eps}
\theoremstyle{thmstylethree}

\begin{document}

\title{Fully Distributed Adaptive Consensus Approach for Economic Dispatch Problem}

\author{
Arnab~Pal*,
Suman~Singha~Roy,
and~Asim~Kumar~Naskar%
\thanks{
\textbf{*Corresponding Author:} Arnab Pal is with the Department of Electrical Engineering,
Vidya Pratishthan's Kamalnayan Bajaj Institute of Engineering \& Technology,
Baramati, Maharashtra 413133, India
(e-mail: arnab.pal@vpkbiet.org).
}%
\thanks{
Suman Singha Roy is with the Department of Electrical Engineering,
IIT Delhi, Delhi, India
(e-mail: Sumansingharoy1998@gmail.com).
}%
\thanks{
Asim Kumar Naskar is with the Department of Electrical Engineering,
National Institute of Technology Rourkela, Odisha, India
(e-mail: naskarasim@gmail.com).
}
}

\maketitle

\begin{abstract}
This research presents a novel approach to solving the economic load dispatch (ELD) problem in smart grid systems by leveraging a multi-agent distributed consensus strategy. The core idea revolves around achieving agreement among generators on their incremental cost values, thereby enabling an optimal allocation of power generation. To enhance convergence and robustness, the study introduces an adaptive coupling weight mechanism within a fully decentralized consensus framework, carefully designed with appropriate initial settings for incremental costs. The proposed distributed control protocol is versatile—it functions effectively in both constrained and unconstrained generator capacity scenarios. Importantly, the methodology ensures that total power generation continuously matches dynamic load demands throughout the dispatch process, maintaining system-wide balance. To accommodate fluctuating and time-varying load profiles, a “dummy node” is incorporated into the network architecture, acting as a flexible proxy for real-time demand changes. The resilience of the method is further evaluated under communication disruptions, specifically by analyzing generator link failures through a switching network topology. Stability of the system is rigorously established using a Lyapunov-based analysis, assuming an undirected and connected communication graph among agents. To validate the practical efficacy of the proposed technique, comprehensive simulations are conducted on the IEEE 30-bus test system within the MATLAB environment, confirming its accuracy, adaptability, and computational efficiency in realistic smart grid conditions.
\end{abstract}


\IEEEpeerreviewmaketitle

\section{Introduction}
The economic load dispatch (ELD) problem aims to determine the most cost-effective generation schedule that meets the total power demand in a smart grid while minimizing overall operational expenses. This challenge has garnered considerable attention in recent research, as it represents a cornerstone issue in modern power system operation. In traditional power systems, centralized optimization techniques—such as genetic algorithms \cite{4} and particle swarm optimization \cite{6,7} have been extensively employed to tackle ELD. These approaches typically rely on a central supervisory controller responsible for coordination, data aggregation, computation, and issuing dispatch commands to individual generators. However, in the context of smart grids—characterized by a large number of distributed energy resources—such centralized architectures become impractical due to scalability and communication bottlenecks \cite{8}. From a network perspective, a smart grid can be modeled as a complex multi-agent system, where nodes represent generating units and edges denote communication or information exchange pathways between them \cite{9}. Under this paradigm, distributed consensus strategies offer a promising alternative: by leveraging only local neighbor-to-neighbor information dictated by the network topology, agents can collaboratively converge toward a globally optimal solution. Distributed consensus control enables autonomous decision-making among agents to achieve a shared objective using limited, localized data—a concept that has attracted significant scholarly interest over the past two decades \cite{10,11}. While prior works such as \cite{30,31} have applied consensus-based methods to ELD, they often depend on global knowledge of the network structure (e.g., Laplacian eigenvalues), undermining full decentralization.  
Recent advances include both discrete-time and continuous-time formulations for distributed optimization \cite{13,14}. For instance, \cite{35} addresses ELD over directed graphs using consensus, while \cite{36,37} propose partially decentralized schemes. A dynamic programming–inspired distributed algorithm is introduced in \cite{16} to optimally allocate load among generators while respecting capacity constraints. Other studies \cite{17,18} present consensus-based dispatch mechanisms, yet they still require partial global network information, falling short of true distribution. In \cite{20}, an external penalty method is used within a consensus framework, with initial conditions set via a distributed allocation rule; however, this approach risks transient violations of generator limits during convergence. Meanwhile, \cite{19} develops a unified framework using a logarithmic barrier function to handle both bounded and unbounded generator capacities, proving stability and demonstrating incremental cost consensus under supply-demand balance—but it still relies on the second smallest eigenvalue of the Laplacian (algebraic connectivity) for controller tuning, which is non-local information. A common limitation across much of this literature is the assumption of fixed coupling weights in the communication graph, which may not reflect real-world dynamics and adaptability requirements. To address these gaps, this paper proposes a fully distributed consensus algorithm for ELD that employs adaptive, time-varying coupling weights eliminating the need for any global network parameters. The method ensures that the total generation cost is minimized precisely when all generators reach consensus on their incremental costs. It is applicable under both constrained and unconstrained generator capacity scenarios, always maintaining supply-demand equilibrium. Furthermore, we introduce the concept of dummy nodes to dynamically accommodate time-varying load demands. To enhance robustness, we also analyze system performance under switching network topologies, simulating communication link failures among generators. Finally, a novel Lyapunov candidate function is constructed to rigorously prove the convergence of the proposed protocol.

The remainder of this paper is organized as follows: Section II reviews essential concepts from graph theory, outlines the single-integrator consensus model, and formulates the ELD problem without generation limits, followed by the design of a distributed consensus protocol. Section III extends the analysis to include generator capacity constraints and addresses both constant and time-varying demand scenarios. Section IV presents numerical simulations based on the IEEE 30-bus test system, implemented in MATLAB, to validate the effectiveness and scalability of the proposed approach.
\section{Preliminaries}
\subsection{Graph Theory}
The interaction topology i.e. the pattern of information exchange among the agents can be represented by a  graph $G_n=(V,E)$, where $V = \left\{ {{v_1},........,{v_N}} \right\}$ is the set of nodes and the set of edges is $E \subseteq \left\{ {\left( {{v_i},{v_j}} \right):{v_i},{v_j} \in V} \right\}$. For $N$ agents the associated adjacency matrix of graph $G_n$ is defined as $\hat A_n = \left[ {{a_{ij}}} \right] \in {R^{N \times N}}$, where $a_{ij}=1$ if $i^{th}$ node collects information from $j^{th}$ , otherwise it will be zero. A graph is called undirected if $a_{ij}=a_{ji}$, otherwise directed. The Laplacian matrix associated with adjacency matrix $\hat A$ can be represented by $L = {\left[ {{l_{ij}}} \right]_{N \times N}}$ and ${l_{ii}} = \sum\limits_{j = 1}^N {{a_{ij}}}$, ${l_{ij}} =  - {a_{ij}},\,\forall i \ne j$. For a given vertex $i$, the neighboring set of nodes can be represented by $N_i$. An undirected graph is connected, if for every pair of nodes, there exists a path, otherwise disconnected. For an undirected connected graph, the following properties are useful \cite{22,23}:
The interaction topology i.e. the pattern of information exchange among the agents can be represented by a  graph $G_n=(V,E)$, where $V = \left\{ {{v_1},........,{v_N}} \right\}$ is the set of nodes and the set of edges is $E \subseteq \left\{ {\left( {{v_i},{v_j}} \right):{v_i},{v_j} \in V} \right\}$. For $N$ agents the associated adjacency matrix of graph $G_n$ is defined as $\hat A_n = \left[ {{a_{ij}}} \right] \in {R^{N \times N}}$, where $a_{ij}=1$ if $i^{th}$ node collects information from $j^{th}$ , otherwise it will be zero. A graph is called undirected if $a_{ij}=a_{ji}$, otherwise directed. The Laplacian matrix associated with adjacency matrix $\hat A$ can be represented by $L = {\left[ {{l_{ij}}} \right]_{N \times N}}$ and ${l_{ii}} = \sum\limits_{j = 1}^N {{a_{ij}}}$, ${l_{ij}} =  - {a_{ij}},\,\forall i \ne j$. For a given vertex $i$, the neighboring set of nodes can be represented by $N_i$. An undirected graph is connected, if for every pair of nodes, there exists a path, otherwise disconnected. For an undirected connected graph, the following properties are useful \cite{22,23}:
\begin{enumerate}
	\item $L = L^T$
	\item $\lambda_1(L) = 0, \lambda_i(L) > 0, \forall i = 2, 3, \dots, N$
	\item $\min_{x \perp \mathbf{1}, x \neq 0} \frac{x^T L x}{x^T x} = \lambda_2(L) \label{eq:aa}$
\end{enumerate}
The Kronecker product of matrices $J\in{R^{2 \times 2}}$ and $V\in{R^{p \times q}}$ defined by, $J \otimes V = {\left[ {\begin{array}{*{20}{c}}
{{j_{11}}V}&{{j_{12}}V}\\
{{j_{21}}V}&{{j_{22}}V}
\end{array}} \right]_{2p \times 2q}}$.
	%\begin{lem}\cite{22}
		%If $a \geqslant 0,\,b \geqslant 0,\,p > 0,\,q > 0$ are real numbers satisfying $\frac{1}{p} + \frac{1}{q} = 1$, then $ab \leqslant \frac{{{a^p}}}{p} + \frac{{{b^q}}}{q}$.
	%\end{lem}
\subsection{Consensus of Multi-agent system}
 \par The consensus among a set of single integrator is described as the following \cite{24,25},
 \begin{equation}
\begin{array}{l}
{{\dot z}_i} ={u}_{i}\\
{u}_{i}= k\sum\limits_{j \in {N_i}} {{a_{ij}}} \left( {{z_j} - {z_i}} \right) =  - k\sum\limits_{j = 1}^N {{l_{ij}}{z_j}}, i=1,2,...,N
\end{array}
 \end{equation}
 \par Here, $z_i$ is the state vector for $i^{th}$ agent, and it can be any physical parameter like the voltage, power, incremental cost, etc, and $k$ is the coupling weight.
 The aggregated form of the above system is given by
\begin{equation*}
\dot z =  - \left( {L \otimes k} \right)z
\end{equation*}
where $z={{\left[ {{z}_{1}}^{T} {{z}_{2}}^{T} \cdots {{z}_{N}}^{T} \right]}^{T}}$
 \par The consensus is achieved by the agents if the following is satisfied,
 \begin{equation*}
\mathop {\lim }\limits_{t \to \infty } \left( {{z_i} - {z_j}} \right) = 0, \mathop {\lim }\limits_{t \to \infty } \left( {{z_i} - {z^{*}}} \right) = 0, \,\forall i,j = 1,2,....,N
 \end{equation*}
Here the consensus state of the participating agents is given by
 \begin{equation*}
 {{z}^{*}}=\sum\limits_{i=1}^{N}{{{p}_{i}}{{z}_{i}}(0)}
 \end{equation*}
 which is a matrix equation in general $\left( z\in {{R}^{n}} \right)$, where ${p}_{i}$ are the elements of the left-eigenvector of the Laplacian matrix corresponding to ${\lambda}_{1}=0$ and ${z}_{i}(0)$ is the initial state values.

%where ${\sigma}_{ij}(t)$ is a positive constant, equals one when information flows from j to i, zero otherwise. By \cite{12}, if there exists an infinite series of uniformly bounded, non-overlapping time intervals, i.e., $\left[ {{t}_{{{p}_{0}}}},{{t}_{{{p}_{f}}}} \right]\cap \left[ {{t}_{{{q}_{0}}}},{{t}_{{{q}_{f}}}} \right]=\phi \,\,\forall p,q= 1,2,....,N $, where ${{t}_{{{p}_{0}}}},{{t}_{{{q}_{0}}}}$ denote the initial and ${{t}_{{{p}_{f}}}},{{t}_{{{q}_{f}}}}$ denote the final time at the ${p}^{th}$ and ${q}^{th}$ graph topology and $\phi$ is a null set; with $\left| {{u}_{i}} \right|\le M \,\,\forall i= 1,2,....,N$, where M is any finite constant; and the union of the graphs across each interval is maintained as a spanning tree, with the topology being fixed in a specific time-interval $\left[ {{t}_{{{k}_{0}}}},{{t}_{{{k}_{f}}}} \right] \,\,\forall k= 1,2,....,N $ then consensus is achieved asymptotically.

\subsection{Representation of Power System Network as Graph Topology}
As an example, the IEEE-30 bus system is considered, as shown in Fig.1(a). In the system, there are generator buses, load buses and they are connected through power line. Generator buses  are important for power dispatch problem. In this work, these buses are considered as a node.
\begin{figure}[!ht]
    \centering

    \begin{subfigure}[b]{0.4\textwidth}
        \centering
        \includegraphics[width=\textwidth]{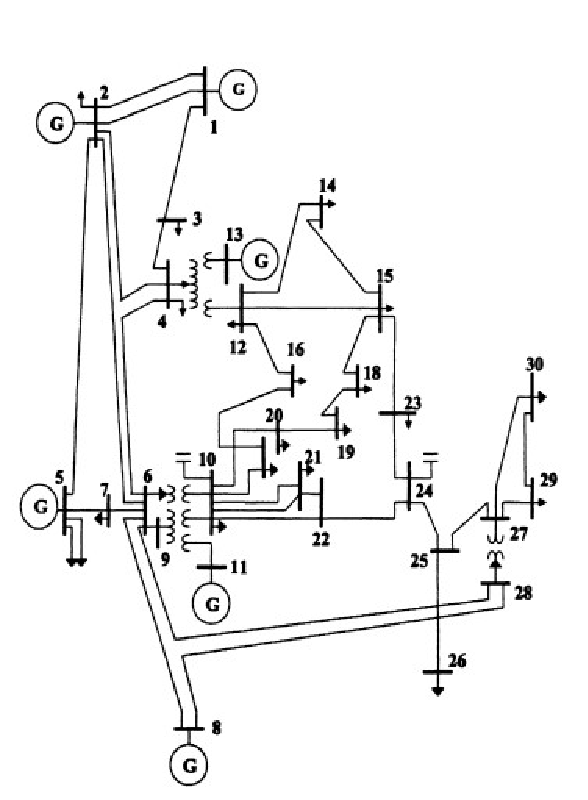}
        \caption{IEEE-30 BUS System}
    \end{subfigure}
    \hspace{-0.08cm} % reduce gap here
    \begin{subfigure}[b]{0.2\textwidth}
        \centering
        \includegraphics[width=\textwidth]{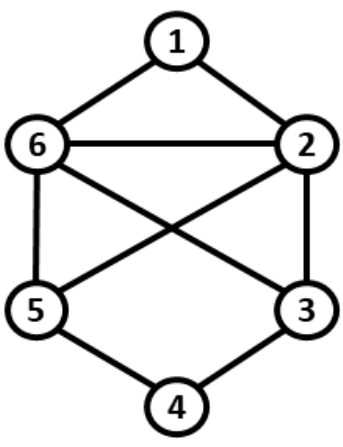}
        \caption{IEEE-30 BUS Network Topology}
    \end{subfigure}

    \caption{Combined view of the IEEE-30 Bus System and its Network Topology}
\end{figure}
The nodes 1, 2, 3, 4, 5, 6 corresponds to generator number 1, 2, 5, 8, 11, 13 of the IEEE-30 BUS system. For this work it is assumed that there exist communication line between generator buses or nodes. A typical example of communication topology is shown in Fig.1(b).
In this section, the economic power dispatch problem is reformulated as the consensus problem in multi-agent systems to achieve the lowest total cost of  generation  in a smart grid, without  constraints on power generation capacity.
\subsection{Problem formulation} \label{ELD}
Consider a distributed optimal power dispatch for smart grids with N generating units, under the assumption of no capacity limitation. The cost function of the $i^{th}$ generating unit is considered \cite{19}, as follows,
\begin{equation}\label{eq:1}
{C_i}\left( {{P_i}} \right) = {a_i} + {b_i}{P_i} + {c_i}{P_i}^2,\,\forall i =1,2,..,N
\end{equation}
\par where $C_i$ is the cost of $i^{th}$ generator, $P_i$ is the active power output of $i^{th}$ generator, and $a_i,b_i,c_i$ are the positive constants. Then the total cost of $N$ generators is $J = \sum\limits_{i = 1}^N {{C_i}\left( {{P_i}} \right)}$. Now, the economic power dispatch problem for $N$ generators can be formulated as,
\begin{equation} \label{eq:2}
	\min_{P_i} \sum_{i=1}^{N} C_i(P_i) \quad \text{s.t.} \quad \sum_{i=1}^{N} P_i = P_D
\end{equation}
\par Here $P_D$ is the total power demand of the smart grid.
\subsection{Solution to power dispatch problem}
Using Lagrange multiplier theory, and KKT condition, the solution to the optimization problem (\ref{eq:2}) is $\frac{{d{J^*}}}{{dP}} = h\textbf{1}$ and $\sum\limits_{i = 1}^N {{P_i} - {P_D} = 0}$, where $h$ is the Lagrange multiplier and $\textbf{1} = {\left[ {\begin{array}{*{20}{c}}
1&.&.&1
\end{array}} \right]^T}$. Considers the incremental cost of the $i^{th}$ generating unit is defined by, ${w_i} = \frac{{d{C_i}\left( {{P_i}} \right)}}{{d{P_i}}}, i=1,2,...,N$. The optimal power dispatch problem is solved if the following two conditions are satisfied : Condition (1): $w_i=w_j, i,j=1,2,...,N$ Condition (2): $\sum\limits_{i = 1}^N {{P_i}}={P_D}$.  
\par This means that consensus on the incremental cost of all generators can solve the optimal power dispatch problem defined by (\ref{eq:2}).
\subsection{Adaptive consensus protocol  for economic dispatch}
Since the final condition of power dispatch problem is maintaining ${w_i} = {w_j}$ at $t \to \infty$. The same endresult is also obtained in consensus of single integrator system \cite{19}.
\par Now inspired by the work on \cite{41}, an adaptive incremental cost consensus algorithm is proposed under undirected connected network graph among the generating nodes. Considering, the incremental cost of $i^{th}$ generator as $w_i = {b_i} + 2{c_i}{P_i}, i=1,2,..., N$.
\begin{theorem}
Under fixed undirected network topology among the generating nodes, the adaptive consensus method described by 
\begin{equation}\label{eq:3}
\begin{array}{l}
{{\dot w}_i} = \sum\limits_{j \in {N_i}} {{a_{ij}}} \left( t \right)\left( {{w_j} - {w_i}} \right)\\
{{\dot a}_{ij}}\left( t \right) = {\beta _{ij}}{\left( {{w_i} - {w_j}} \right)^2},\;j \in {N_i}
\end{array}
\end{equation}
where $\beta_{ij}=\beta_{ji}$ are positive scalars, $a_{ij}(t)$ is the entries of the adjacency matrix and can maintain
stability when applied to solve the economic load dispatch
problem expressed by Conditions 1 and 2.
\end{theorem}
\textit{Proof}: 
To show the balance condition between supply and demand i.e. condition 2, let us assume the average consensus point is $\bar w\left( t \right) = \frac{1}{N}\sum\limits_{i = 1}^N {{w_i}\left( t \right)}$. Since ${a_{ij}}\left( t \right) = {a_{ji}}\left( t \right),\;j \in {N_i}$, then the following can be written,
\begin{equation}\label{eq:ff}
\dot {\bar {w}} = \frac{1}{N}\sum\limits_{i = 1}^N {\sum\limits_{j \in {N_i}} {{a_{ij}}} \left( t \right)\left( {{w_j} - {w_i}} \right) = 0}
\end{equation}
\par Using ${w_i}\left( t \right) = {b_i} + 2{c_i}{P_i}\left( t \right)$, we can write $\bar w\left( t \right) = \mu \sum\limits_{i = 1}^N {\left( {\frac{{{b_i}}}{{2{c_i}}} + {P_i}\left( t \right)} \right)}$, where $\mu  = \frac{{\sum\limits_{i = 1}^N {2{c_i}} }}{N}$. Now ${\dot {\bar w}}\left( t \right) = \mu \sum\limits_{i = 1}^N {{{\dot P}_i}} \left( t \right) = 0$, which indicates that, $\sum\limits_{i = 1}^N {{P_i}\left( t \right) = } \sum\limits_{i = 1}^N {{P_i}\left( 0 \right) = } {P_D}$.
\par According to the above-mentioned analysis, the demand and supply of power in the considered smart grid are always balanced $\forall t \ge 0$. 
\par Now for condition 1, let ${e_i} = {w_i} - \bar w,\,\forall i = 1,...,N$, therefore the incremental cost consensus for all generating units will be achieved if $\mathop {\lim }\limits_{t \to \infty } \left\| {{e_i}\left( t \right)} \right\| = 0,\,\forall i = 1,...,N$. By using (\ref{eq:ff}) one can obtain the following,
\begin{equation}
{{\dot e}_i} = \sum\limits_{j \in {N_i}} {{a_{ij}}} \left( t \right)\left( {{e_j} - {e_i}} \right)
\end{equation}
with ${{\dot a}_{ij}}\left( t \right) = {\beta _{ij}}{\left( {{e_i} - {e_j}} \right)^2},\;j \in {N_i}$. Now the aggregated form of the error dynamics can be written as,
\begin{equation}\label{eq:gg}
\dot e =  - L\left( t \right)e\left( t \right)
\end{equation}
where $e = {\left[ {\begin{array}{*{20}{c}}
{{e_1}\left( t \right)}&.&.&{{e_N}\left( t \right)}
\end{array}} \right]^T}$, next to show the stability analysis the following Lyapunov energy function is considered for system (\ref{eq:gg}),
\begin{equation}
V = \frac{1}{2}e{\left( t \right)^T}e\left( t \right) + \sum\limits_{i = 1}^N {\sum\limits_{j \in {N_i}} {\frac{1}{{4{\beta _{ij}}}}} } {\left( {{\theta _{ij}} - {a_{ij}}\left( t \right)} \right)^2}
\end{equation}
where ${\theta _{ij}} = {\theta _{ji}} \ge {\theta _0}$ if ${a_{ij}}\left( 0 \right) > 0$ otherwise ${\theta _{ij}} = {\theta _{ji}} = 0$, $\theta_0$ is a given positive constant. Using (\ref{eq:3}), the following can be obtained,
\begin{equation}\label{eq:hh}
\dot V =  - e{\left( t \right)^T}L\left( t \right)e\left( t \right) - \frac{1}{2}\sum\limits_{i = 1}^N {\sum\limits_{j \in {N_i}} {\left( {{\theta _{ij}} - {a_{ij}}\left( t \right)} \right){{\left( {{e_j} - {e_i}} \right)}^2}} }
\end{equation}
Moreover $e{\left( t \right)^T}L\left( t \right)e\left( t \right) = \frac{1}{2}\sum\limits_{i = 1}^N {\sum\limits_{j \in {N_i}} {{a_{ij}}\left( t \right){{\left( {{e_j} - {e_i}} \right)}^2}} }$. So, the following can be derived using the above analysis and (\ref{eq:hh}),
\begin{equation}
\dot V =  - \frac{1}{2}\sum\limits_{i = 1}^N {\sum\limits_{j \in {N_i}} {{\theta _{ij}}{{\left( {{e_j} - {e_i}} \right)}^2}} }
\end{equation}
Since ${\theta _{ij}} = {\theta _{ji}} \ge {\theta _0}$ , the following can be written,
\begin{equation}\label{eq:ii}
\dot V \le  - \frac{{{\theta _0}}}{2}\sum\limits_{i = 1}^N {\sum\limits_{j \in {N_i}} {{{\tilde a}_{ij}}{{\left( {{e_j} - {e_i}} \right)}^2}} }
\end{equation}
where $\tilde a_{ij}=1$ for $j \in {N_i}$ and $\tilde a_{ij}=0$ otherwise, for all $i=1,2,....,N$. Then, (\ref{eq:ii}) can be rewritten as 
\begin{equation}\label{eq:jj}
\dot V \le - {\theta _0}e{\left( t \right)^T}\tilde Le\left( t \right)
\end{equation}
where $\tilde L = {\left[ {{\tilde l_{ij}}} \right]_{N \times N}}$ with ${\tilde l_{ij}} =  - {\tilde a_{ij}}$ and ${\tilde l_{ii}} = \sum\limits_{j = 1,j \ne i}^N {{\tilde a_{ij}}}$ for all $i,j=1,2,...,N$.  $\tilde L$ is the Laplacian matrix for the undirected network topology. Now using the property of Laplacian matrix for undirected graph topology, the following can be obtained,
\begin{equation}\label{eq:kk}
\dot V \le  - {\lambda _2}\left( \tilde L \right){\theta _0}e{\left( t \right)^T}e\left( t \right)
\end{equation}
From (\ref{eq:kk}), it can be concluded that $\dot V$ is bounded, and $\dot V=0$ only if $ e=0$ because $\tilde L$ is always positive definite. Now based on the above mentioned analysis  one can show $\mathop {\lim }\limits_{t \to \infty }  e \to 0$ i.e. $\mathop {\lim }\limits_{t \to \infty } { e_i} \to 0$. From these results, it can be said that a consensus of incremental cost is reached using the proposed fully distributed protocol.
\hfill $\blacksquare$\\
\par The process of distributed adaptive consensus for solving the EDP is described in the following pseudo-code at the ${i}^{th}$ generating node.
\vspace{5pt}
\hrule
\begin{enumerate}
    \item Collect information ${{w}_{j}}\,\forall j\in {{N}_{i}}$
    \item Initialize Active Power Outputs $\left( {{P}_{i}}\left( 0 \right) \right)$
    \item Calculate the Incremental Costs $\left({\tilde w_i}\left( 0 \right) = I{{C}_{i}}\left( 0 \right)={{b}_{i}}+2{{c}_{i}}{{P}_{i}}\left( 0 \right) \right)$
    \item Calculate ${\sigma _i} = {\beta _{ij}}{\left( {{w_i} - {w_j}} \right)^2}$ then  by integrating ${\sigma _i}$ evaluate $a_{ij}(t)$ \label{algo1}
    \item Apply control input ${u_i} = \sum\limits_{j \in {N_i}} {{a_{ij}}} \left( t \right)\left( {{w_j} - {w_i}} \right)$
    \item \textbf{If} ${{u}_{i}}\approx0$ \textbf{then}
    \subitem Stop
    \item \textbf{Else}
    \subitem Goto \ref{algo1}
\end{enumerate}
\hrule
%\vspace{5pt}
%\captionof{figure}{Fully Distributed Adaptive Consensus for solving EDP}
\begin{remark}
In contrast to the centralized algorithms in the existing literature, a fully distributed economic load dispatch strategy is designed based on consensus in multi-agent systems. It is important to note that the graph is considered to be connected and the proposed approach is not affected by any global information like the second smallest eigenvalue of the Laplacian matrix or a number of agents, which makes this approach fully distributed. The optimal  economic power dispatch in smart grids can then be solved using distributed control and the adaptive weights $a_{ij}$ for $i,j=1,2,....,N$ will reach some positive constants. Furthermore, from the definition of incremental cost $w_i$, the output power of $N$ generators will converge to some positive constants if the state consensus among $w_i$, $i=1,2,...,N$ is reached.
\end{remark}
\section{Optimal Power Dispatch  With Capacity Limitations}
A fully distributed economic power dispatch algorithm is designed in this section to minimize the cost function of generators in a smart grid with capacity constraints by using adaptive consensus in multi-agent systems.
\subsection{Optimal Power Dispatch with Constant Demand}
\subsubsection{Problem Formulation}
Consider the same optimization problem defined by (\ref{eq:2}), with generator capacity limitations
\begin{equation}\label{eq:8}
\begin{split}
\mathop {\min }\limits_{{P_i}} \sum\limits_{i =1 }^N {{C_i}\left( {{P_i}} \right) = } \mathop {\min }\limits_{{P_i}} \sum\limits_{i = 1}^N {{a_i} + {b_i}{P_i}}  + {c_i}{P_i}^2\\s.t\, (i) \sum\limits_{i = 1}^N {{P_i}}  = {P_D},\, (ii) {P_i}^{\min } \le {P_i} \le {P_i}^{\max }
\end{split}
\end{equation}
\par where each generating unit is constrained by the maximum limit i.e. $P_i^{\max}$ and minimum limit i.e. $P_i^{\min}$ for $i=1,2,....,N$.
\par The mixed constrained optimization problem (\ref{eq:8})can be converted to an equally constrained convex optimization problem using the logarithm barrier function \cite{40}. Using logarithm barrier the cost function can be modified to the following
\begin{align*}
	C_i^\delta (P_i) &= C_i(P_i) - \delta \left( \ln(P_i - P_i^{\min}) + \ln(P_i^{\max} - P_i) \right) \\
	C_i^\delta (P_i) &= a_i + b_i P_i + c_i P_i^2 - \delta \left( \ln(P_i - P_i^{\min}) + \ln(P_i^{\max} - P_i) \right)
\end{align*}
\par here $\delta$ is a small value. It can be checked that, ${C_i}^\delta \left( {{P_i}} \right)$ is a convex function for ${P_i} \in \left[ {{P_i}^{\min },{P_i}^{^{\max }}} \right]$. Then, the optimization problem defined by (\ref{eq:8}) can be converted to the following constrained optimization problem,
\begin{equation}\label{eq:9}
\begin{split}
  & \underset{{{P}_{i}}}{\mathop{\min }}\,\sum\limits_{i=1}^{N}{{{C}_{i}}^{\delta }\left( {{P}_{i}} \right)\ } \\ 
 & s.t.\sum\limits_{i=1}^{N}{{{P}_{i}}={{P}_{D}}} \\ 
\end{split}
\end{equation}
Considers the incremental cost of the $i^{th}$ generating unit is defined by, ${{\tilde w}_i} = \frac{{d{C_i}^\delta \left( {{P_i}} \right)}}{{d{P_i}}}, i=1,2,...,N$. The optimal power dispatch problem is solved if the following two conditions are satisfied \cite{19,27}: Condition (3): ${\tilde w}_i={\tilde w}_j, i,j=1,2,...,N$ Condition (4): $\sum\limits_{i = 1}^N {{P_i}}={P_D}$.
%%%%%%%%%%%%%%%%%%%%%%%%%%%%%%%%%%%%%%%%%
%%%%%%%%%%%
%%%%%%%%%%%%%%%%%%%%%%%%%%%%%%%%%%%%%%%%%%%
\subsubsection{Distributed consensus Protocol  with the capacity limit for optimal power dispatch}
The proposed distributed adaptive consensus algorithm is proven to be stable by the following theorem.
\begin{theorem}
Given the generator capacity limit, the proposed adaptive consensus protocol described by \begin{equation}\label{eq:10}
\begin{array}{l}
{{\dot P}_i} = \sum\limits_{j \in {N_i}} {{a_{ij}}} \left( t \right)\left( {{{\tilde w}_j} - {{\tilde w}_i}} \right)\\
{{\dot a}_{ij}}\left( t \right) = {\beta _{ij}}{\left( {{{\tilde w}_i} - {{\tilde w}_j}} \right)^2},\;j \in {N_i}
\end{array}
\end{equation}
 is stable in solving the modified economic load dispatch problem defined by (\ref{eq:9}); where
 \begin{equation}\label{eq:ddd}
     {{{\tilde w}_i}}= {b_i} + 2{c_i}{P_i} - \frac{\delta }{{{P_i} - {P_i}^{\min }}} + \frac{\delta }{{{P_i}^{\max } - {P_i}}}
 \end{equation}
 where $\beta_{ij}=\beta_{ji}$ are positive scalars, $a_{ij}(t)$ is the entries of the adjacency matrix.
\end{theorem}
\textit{Proof}:
 The incremental cost consensus is achieved if $\mathop {\lim }\limits_{t \to \infty } \left( {{\tilde w_i} - {\tilde w_j}} \right)  = 0$.  From equation (\ref{eq:ddd}), one can evaluate the following
\begin{equation}\label{eq:11}
\frac{{d{{\tilde w}_i}}}{{d{P_i}}} = 2{c_i} + \frac{\delta }{{{{\left( {{P_i} - {P_i}^{\min }} \right)}^2}}} + \frac{\delta }{{{{\left( {{P_i}^{\max } - {P_i}} \right)}^2}}} > 0    
\end{equation}
Let \begin{equation*}
\frac{{d{{\tilde w}_i}}}{{d{P_i}}} = \gamma
\end{equation*} where $\gamma$ is always a positive value. Now, \[\frac{{d{{\tilde w}_i}}}{{dt}} = \frac{{d{{\tilde w}_i}}}{{d{P_i}}}\frac{{d{P_i}}}{{dt}} = \gamma \frac{{d{P_i}}}{{dt}}\] 
Let us assume the average consensus point is $\hat w\left( t \right) = \frac{1}{N}\sum\limits_{i = 1}^N {{\tilde {w}_i}\left( t \right)}$. Since ${a_{ij}}\left( t \right) = {a_{ji}}\left( t \right),\;j \in {N_i}$, then the following can be written,
\begin{equation}\label{eq:ll}
\dot {\hat {w}} = \frac{\gamma}{N}\sum\limits_{i = 1}^N {\sum\limits_{j \in {N_i}} {{a_{ij}}} \left( t \right)\left( {{\tilde {w}_j} - {\tilde {w}_i}} \right) = 0}
\end{equation}
\par Now for condition 3, let ${\hat{e}_i} = {\tilde {w}_i} - \bar w,\,\forall i = 1,...,N$, therefore the modified incremental cost consensus for all generating units will be achieved if $\mathop {\lim }\limits_{t \to \infty } \left\| {{\hat{e}_i}\left( t \right)} \right\| = 0,\,\forall i = 1,...,N$. By using (\ref{eq:ll}) one can obtain the following,
\begin{equation}
{\dot {\hat {e}}_i} = \gamma\sum\limits_{j \in {N_i}} {{a_{ij}}} \left( t \right)\left( {{\hat{e}_j} - {\hat{e}_i}} \right)
\end{equation}
with ${{\dot a}_{ij}}\left( t \right) = {\beta _{ij}}{\left( {{\hat{e}_i} - {\hat{e}_j}} \right)^2},\;j \in {N_i}$. Now the aggregated form of the error dynamics can be written as,
\begin{equation}\label{eq:gg}
\dot {\hat{e}} =  -\gamma L\left( t \right)\hat{e}\left( t \right)
\end{equation}
where $\hat{e} = {\left[ {\begin{array}{*{20}{c}}
{{\hat{e}_1}\left( t \right)}&.&.&{{\hat{e}_N}\left( t \right)}
\end{array}} \right]^T}$, next to show the stability analysis the following Lyapunov energy function is considered for system (\ref{eq:gg}),
\begin{equation}
V = \frac{1}{2}\hat{e}{\left( t \right)^T}\hat{e}\left( t \right) + \sum\limits_{i = 1}^N {\sum\limits_{j \in {N_i}} {\frac{\gamma}{{4{\beta _{ij}}}}} } {\left( {{\theta _{ij}} - {a_{ij}}\left( t \right)} \right)^2}
\end{equation}
where ${\theta _{ij}} = {\theta _{ji}} \ge {\theta _0}$ if ${a_{ij}}\left( 0 \right) > 0$ otherwise ${\theta _{ij}} = {\theta _{ji}} = 0$, $\theta_0$ is a given positive constant. Using (\ref{eq:10}), the following can be obtained,
\begin{equation}\label{eq:hh}
\dot V =  - \gamma \hat{e}{\left( t \right)^T}L\left( t \right)\hat{e}\left( t \right) - \frac{\gamma}{2}\sum\limits_{i = 1}^N {\sum\limits_{j \in {N_i}} {\left( {{\theta _{ij}} - {a_{ij}}\left( t \right)} \right){{\left( {{e_j} - {e_i}} \right)}^2}} }
\end{equation}
 Now using the property of the Laplacian matrix for undirected graph topology and the derivation in Theorem 1, the following can be obtained,
\begin{equation}\label{eq:kk}
\dot V \le  - \gamma{\lambda _2}\left( \tilde L \right){\theta _0}\hat{e}{\left( t \right)^T}\hat{e}\left( t \right)
\end{equation}
\par From (\ref{eq:kk}), it can be concluded that $\dot V$ is bounded, and $\dot V=0$ only if $\hat e=0$ because the constant term inside the inequality (\ref{eq:kk}) is always positive definite. Now based on the above mentioned analysis  one can show $\mathop {\lim }\limits_{t \to \infty } \hat e \to 0$ i.e. $\mathop {\lim }\limits_{t \to \infty } {\hat e_i} \to 0$. From these results, it can be said that a consensus of modified incremental cost is reached using the proposed fully distributed protocol. For condition 4, from (\ref{eq:10})  one can write $\sum\limits_{i = 1}^N {{{\dot P}_i}\left( t \right) = 0}$, which indicates that $\sum\limits_{i = 1}^N {{P_i}} \left( t \right) = \sum\limits_{i = 1}^N {{P_i}} \left( 0 \right) = {P_D}$. If the initial condition i.e. $\sum\limits_{i = 1}^N {{P_i}} \left( 0 \right) = {P_D}$ is met, then the balance between the supply and demand of powers is maintained.
\hfill $\blacksquare$\\
\subsection{Optimal power dispatch with Time-Varying Demand}
In this subsection to cater to the time-varying demand, a new concept of virtual nodes or dummy nodes has been introduced, as shown in Fig. 2., to solve the optimal power dispatch problem with generator capacity limits. In this case, the optimization problem  is given by 
\begin{equation}\label{eq:cc}
\begin{array}{lll}
     & \mathop {\min }\limits_{{P_i}} \sum\limits_{i =1 }^{2N} {{C_i}\left( {{P_i}} \right) = } \mathop {\min }\limits_{{P_i}} \sum\limits_{i = 1}^{2N} {{a_i} + {b_i}{P_i}}  + {c_i}{P_i}^2\\
     & s.t\, (i) \sum\limits_{i = 1}^{2N} {{P_i}}  = {P_D},\, (ii) {P_i}^{\min } \le {P_i} \le {P_i}^{\max } ,\\
     &\forall i= 1,2,....,N, (iii){P_i} \ge 0 ,\,\forall i= N+1,N+2,....,2N 
\end{array}
\end{equation}
\begin{figure}[!ht]
	\centering
	%\rule{12.8cm}{7.2cm}
	\includegraphics[scale=0.3]{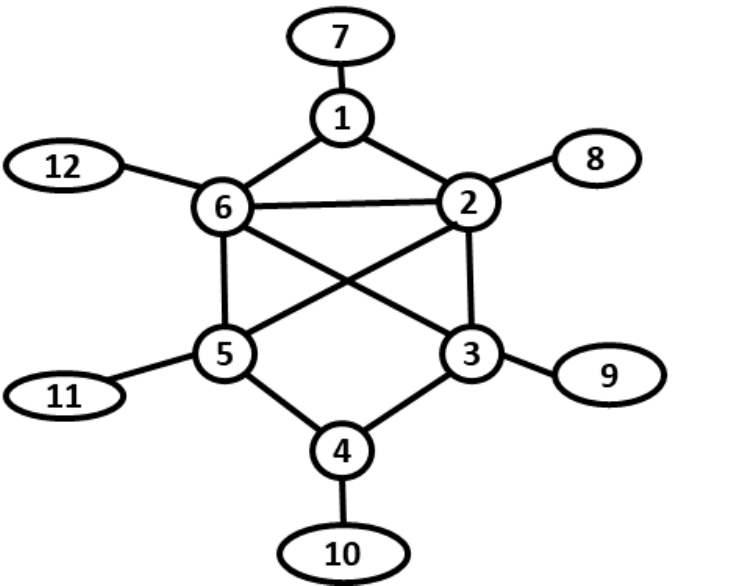}
	\caption{Network Topology}
\end{figure}
here the the total number of generators is $2N$ of which $i=1,....,N$ is considered as real generators nodes while $i=N+1,....,2N$ is considered as dummy generating nodes, and the cost is assumed to be high for dummy nodes. So, for optimization with generator capacity, the modified  cost function for all generator nodes i.e. $i=1,2,..,2N$ are considered
\begin{equation}
\begin{split}\label{eq:ccc} 
 & {{C}_{i}}^{\delta }\left( {{P}_{i}} \right)=\sum\limits_{i=1}^{2N}{\left( {{a}_{i}}+{{b}_{i}}{{P}_{i}}+{{c}_{i}}{{P}_{i}}^{2} \right)}\\
 &-\delta \left( \sum\limits_{i=1}^{N}{\left( \ln \left( {{P}_{i}}-{{P}_{i}}^{\min } \right)+\ln \left( {{P}_{i}}^{\max }-{{P}_{i}} \right) \right)}\right)\\
 &-\delta \sum\limits_{i = N + 1}^{2N} {\ln \left( {{P_i}} \right)}\\  
\end{split}
\end{equation}
%Then for dummy nodes i.e. $i=N+1,....,2N$, assuming there are no generator limits, the modified cost is 
%\begin{equation*}
%{C_i}^\delta \left( {{P_i}} \right) = \sum\limits_{i = N+1}^{2N} {{a_i} + {b_i}{P_i} + {c_i}{P_i}^2
%\end{equation*}
%So, the modified constrained optimization problem is the following,
%\begin{equation}\label{eq:bb}
%\begin{split}
  %& \underset{{{P}_{i}}}{\mathop{\min }}\,\sum\limits_{i=1}^{2N}{{{C}_{i}}^{\delta }\left( {{P}_{i}} \right)\ } \\ 
 %& s.t.\sum\limits_{i=1}^{2N}{{{P}_{i}}={{P}_{D}}} \\ 
%\end{split}
%\end{equation}
In this case, the proposed adaptive distributed protocol is the same as (\ref{eq:10}) with new initial conditions satisfying ${P_i}\left( 0 \right) + {P_{N + i}}\left( 0 \right) = {P_{{D_i}}},\,\forall i = 1,...,N$ is the local demand of the $i^{th}$ generator bus. To have a negligible demand drop over the dummy node, $c_i$ for the dummy nodes, i.e. for $i=N+1,...,2N$, is considered to be large in comparison to real generating nodes, moreover the $a_i, b_i$ for $i=N+1,...,2N$ considered to be very low. 
\begin{theorem}
The following  adaptive consensus protocol for $\forall i \in \left\{ {1,2,....,2N} \right\}$
\begin{equation}
\begin{array}{l}
{{\dot P}_i} = \sum\limits_{j \in {N_i}} {{a_{ij}}} \left( t \right)\left( {{{\tilde w}_j} - {{\tilde w}_i}} \right)\\
{{\dot a}_{ij}}\left( t \right) = {\beta _{ij}}{\left( {{{\tilde w}_i} - {{\tilde w}_j}} \right)^2},\;j \in {N_i}
\end{array}
\end{equation}
can effectively handle the modified economic load dispatch problem with time-varying demand described in \ref{eq:ccc} while considering the generator capacity constraint, ensuring stability in the system; where 
\begin{equation}\label{eq:mm}
\begin{array}{l}
    \tilde{{w}_{i}}=I{C_i}^\delta  = {b_i} + 2{c_i}{P_i} - \frac{\delta }{{{P_i} - {P_i}^{\min }}} + \frac{\delta }{{{P_i}^{\max } - {P_i}}},\\ \forall i=1,2,....,N\\
    \tilde{{w}_{i}}=I{C_i}^\delta = {b_i} + 2{c_i}{P_i}-\frac{\delta }{{P_i}},\,\forall i=N+1,N+2,....,2N
\end{array}
    \end{equation} is the modified incremental cost.
\end{theorem}
\textit{Proof:} 
The incremental cost consensus is achieved if $\mathop {\lim }\limits_{t \to \infty } \left( {{\tilde w_i} - {\tilde w_j}} \right)  = 0$.  From equation (\ref{eq:mm}), one can evaluate the following for $i=1,...,N$. 
\begin{equation*}
\frac{{d{{\tilde w}_i}}}{{d{P_i}}} = 2{c_i} + \frac{\delta }{{{{\left( {{P_i} - {P_i}^{\min }} \right)}^2}}} + \frac{\delta }{{{{\left( {{P_i}^{\max } - {P_i}} \right)}^2}}} > 0    
\end{equation*}
for $i=N+1,....,2N$
\begin{equation*}
\frac{{d{{\tilde w}_i}}}{{d{P_i}}} = 2{c_i} + \frac{\delta }{{{P_i}^2}} > 0
\end{equation*}
So, for $i=1,.....,2N$ let consider the following,
\begin{equation*}
\frac{{d{{\tilde w}_i}}}{{d{P_i}}} = \gamma
\end{equation*}
where $\gamma$ is always a positive value. Now, \[\frac{{d{{\tilde w}_i}}}{{dt}} = \frac{{d{{\tilde w}_i}}}{{d{P_i}}}\frac{{d{P_i}}}{{dt}} = \gamma \frac{{d{P_i}}}{{dt}}\] 
Then assume the average consensus point is $\hat w\left( t \right) = \frac{1}{2N}\sum\limits_{i = 1}^{2N} {{\tilde {w}_i}\left( t \right)}$. Since ${a_{ij}}\left( t \right) = {a_{ji}}\left( t \right),\;j \in {N_i}$, so the following can be written,
\begin{equation*}
\dot {\hat {w}} = \frac{\gamma}{2N}\sum\limits_{i = 1}^{2N} {\sum\limits_{j \in {N_i}} {{a_{ij}}} \left( t \right)\left( {{\tilde {w}_j} - {\tilde {w}_i}} \right) = 0}
\end{equation*}
Now the stability proof of the proposed approach exactly follows the proof of Theorem 2 provided the topologies remain constant. \hfill $\blacksquare$
\par The process of distributed adaptive consensus for solving the EDP is described in the following pseudo-code at the ${i}^{th}$ generating node.
\vspace{8pt}
\hrule
\begin{enumerate}
    \item Collect information ${{\Tilde{w}}_{j}}\,\forall j\in {{N}_{i}}$
    \item Initialize Active Power Outputs $\left( {{P}_{i}}\left( 0 \right) \right)$
    \item Calculate the Incremental Costs $\left({\tilde w_i} =  I{{C}_{i}}={{b}_{i}}+2{{c}_{i}}{{P}_{i}} \right) - \delta \left( {\ln \left( {{P_i} - {P_i}^{\min }} \right) + \ln \left( {{P_i}^{^{\max }} - {P_i}} \right)} \right)$ \label{algo2}
    \item Calculate ${\sigma _i} = {\beta _{ij}}{\left( {{\Tilde{w}_i} - {\Tilde{w}_j}} \right)^2}$ then  by integrating ${\sigma _i}$ evaluate $a_{ij}(t)$ \label{algo2}
    \item Apply control input ${u_i} = \sum\limits_{j \in {N_i}} {{a_{ij}}} \left( t \right)\left( {{\Tilde{w}_j} - {\Tilde{w}_i}} \right)$
    \item \textbf{If} ${{u}_{i}}\approx0$ \textbf{then}
    \subitem Stop
    \item \textbf{Else}
    \subitem Goto \ref{algo2}
\end{enumerate}
\hrule
\vspace{5pt}
\begin{remark}
So far, in this section, an adaptive consensus protocol is proposed and the incremental cost is modified due to the capacity of the generator limits consideration. The proposed protocol successfully solves the optimal power dispatch problem and this approach is different from the approach without a generator capacity limit. The optimal solution of the problem (\ref{eq:9}) under a very small $\delta$ will converge to the global optimal solution of the problem (\ref{eq:9}). Moreover, in the case of time-varying demand, the power demand information is no longer global information, it is considered as local information i.e. $P_{D_i}$.
\end{remark}
\subsection{Optimum Power Dispatch with Generator Communication Link Failure}
In the context of consensus-based power dispatch systems,  the communication link between the generators is essential. However, there may be instances where the communication link fails due to a variety of reasons, including hardware malfunction, network congestion, or environmental factors. In terms of network topology, this entails the addition or deletion of a specific number of edges or switching of topology.  In \cite{39} it has been investigated that it is still possible to attain consensus with switching topology.

\par In this case the graph $G$ is a discreet state, moreover it is described by  a collection of finite numbers of connected graphs, and represented as  ${G_n} = \left\{ G \right\}$. Under generator capacity limit, the consensus error dynamics described by (\ref{eq:gg}) can be rewritten as,
\begin{equation}\label{eq:bbb}
\dot {\hat {e}} =  - \gamma {L_{{G_k}}}\left( t \right)\hat e\left( t \right)
\end{equation}
where $k = s\left( t \right),\,{G_k} \in {G_n}$, and $s(t)$ is a switching signal. $L_{G_k}$ is the Laplacian matrix for $G_k$.
\begin{theorem}
Given the generator capacity limit, the proposed adaptive consensus protocol described by (\ref{eq:10}) is stable in solving the modified economic load dispatch problem defined by (\ref{eq:9}) for the switching topology $G_n$.
\end{theorem}    
\textit{Proof:}
The stability of the switching system can be verified by the following Lyapunov function,
\begin{equation*}
V = \frac{1}{2}\hat{e}{\left( t \right)^T}\hat{e}\left( t \right) + \sum\limits_{i = 1}^N {\sum\limits_{j \in {N_i}} {\frac{\gamma}{{4{\beta _{ij}}}}} } {\left( {{\theta _{ij}} - {a_{ij}}\left( t \right)} \right)^2}
\end{equation*}
Next using the derivation of Theorem 2 in subsection C, the following can be written
\begin{equation}
\dot V \le  - \gamma {\theta _0}\hat e{\left( t \right)^T}{L_{{G_k}}}\hat e\left( t \right) \le 0
\end{equation}
This guarantees that $V(.)$ is a suitable Lyapunov function for the switching system (\ref{eq:bbb}), moreover consensus error vector $\hat e$ is exponentially stable. After this, from Theorem 2 it can also be verified that the economic load dispatch problem is solved.  
\begin{remark}
In this part, switching topology is used to study what happens when a generator link fails. Also, if one generator in the network is completely cut off from the rest of the generators, meaning it is not getting or sending any information, then the disconnected generator stops updating its incremental cost and stops following the other generator, which can be verified by equation (\ref{eq:10}).
\end{remark}
\section{Simulation Examples}
\subsection{IEEE 30 Bus System without Generating Limits with Fixed Topology} \label{5.1}
 Simulation example is provided in this part to demonstrate the validity of the suggested distributed economic power dispatch algorithms. Here an optimal economic dispatch of IEEE-30 bus system \cite{19} without generating limits is demonstrated using the proposed approach. The parameters of the cost function are taken from Table 1, and the $\delta$ is assumed to be 10. The test model of the IEEE-30 bus system \cite{shahidehpour2004communication} can be seen in Fig1(a). and the topology among the generating units in the undirected and connected network is shown in Fig1(b).
 \begin{table}[h!]
 	\centering
 	\caption{Parameters of cost function in IEEE-30 bus system}
 	\label{tab:1}
 	\begin{tabular}{|l|c|c|c|c|c|}
 		\hline
 		Unit (Gen No.) & $a_i$ & $b_i$ & $c_i$ & $P_i^{\min}$ & $P_i^{\max}$ \\ \hline
 		1 (1) & 0 & 2.00 & 0.003750 & 50 & 200 \\ \hline
 		2 (2) & 0 & 1.75 & 0.001750 & 100 & 400 \\ \hline
 		3 (5) & 0 & 1.00 & 0.062500 & 15 & 50 \\ \hline
 		4 (8) & 0 & 3.25 & 0.008324 & 10 & 35 \\ \hline
 		5 (11) & 0 & 3.00 & 0.025000 & 10 & 30 \\ \hline
 		6 (13) & 0 & 3.00 & 0.025000 & 12 & 40 \\ \hline
 	\end{tabular}
 \end{table}
\par Let us assume the power demand $P_D=531.34 MW$, and the initial generating powers of all units are ${P_0} = \left\{ {133.3,287.9,40,38,15,17} \right\}$. Now the consensus of  incremental cost of all the generating units can be seen by Fig 3 and from Fig 4. ,it can be verified that the adaptive  weights  settle down some positive fixed value. Moreover the Fig. 5. verifies the active power outputs.
%Finally, the optimal output power of generators is 139.9699, 371.7272, 16.4345, 10.0137, 10.0111, 12.0091 respectively.

\begin{figure}[!ht]
	\centering
    \captionsetup{justification=centering}
	%\rule{12.8cm}{7.2cm}
	\includegraphics[width=0.99\linewidth,height=40mm, trim=10 0 70 20,clip]{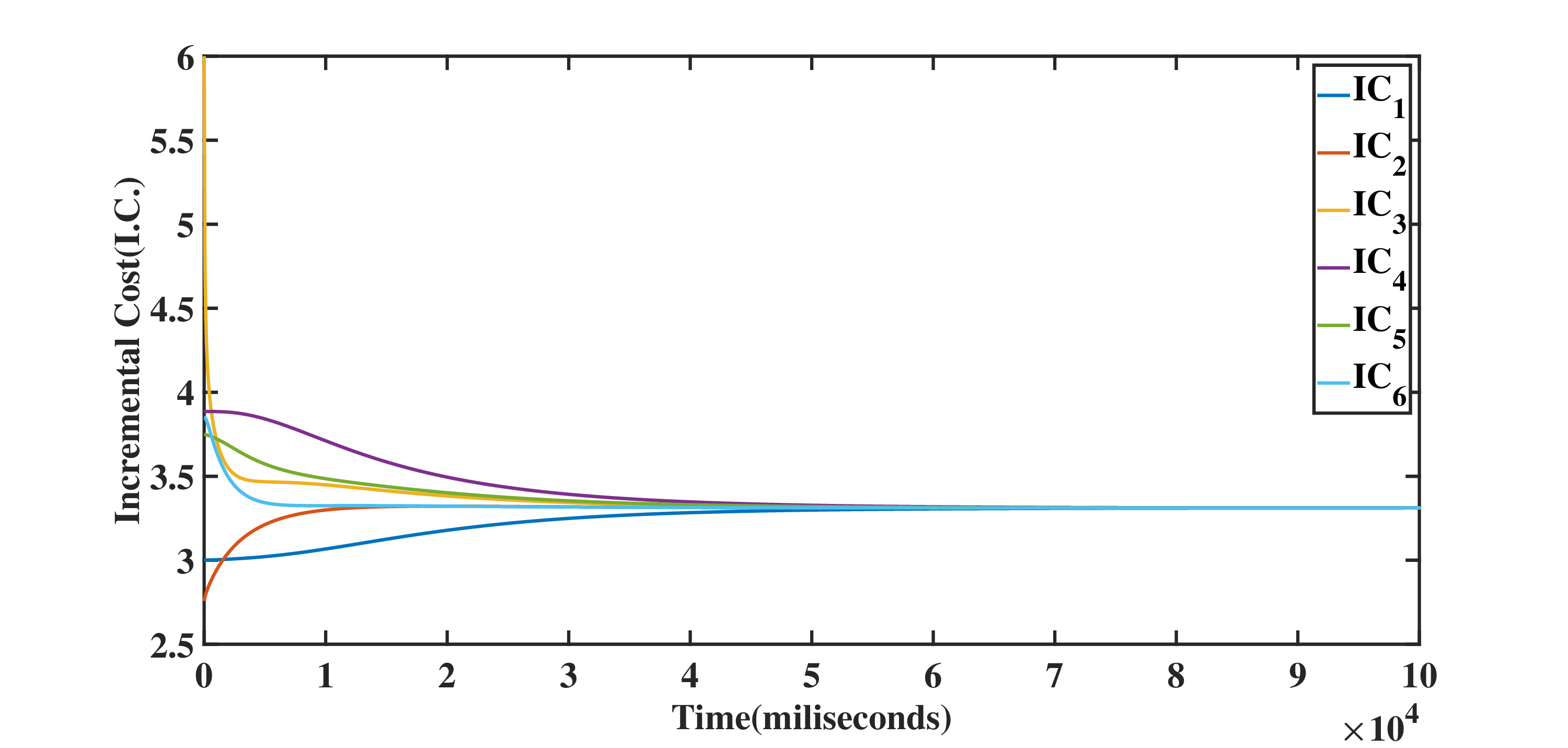}
	\caption{Incremental cost consensus}
\end{figure}

\begin{figure}[!ht]
	\centering
    \captionsetup{justification=centering}
	%\rule{12.8cm}{7.2cm}
	\includegraphics[width=0.99\linewidth,height=40mm, trim=10 0 70 20,clip]{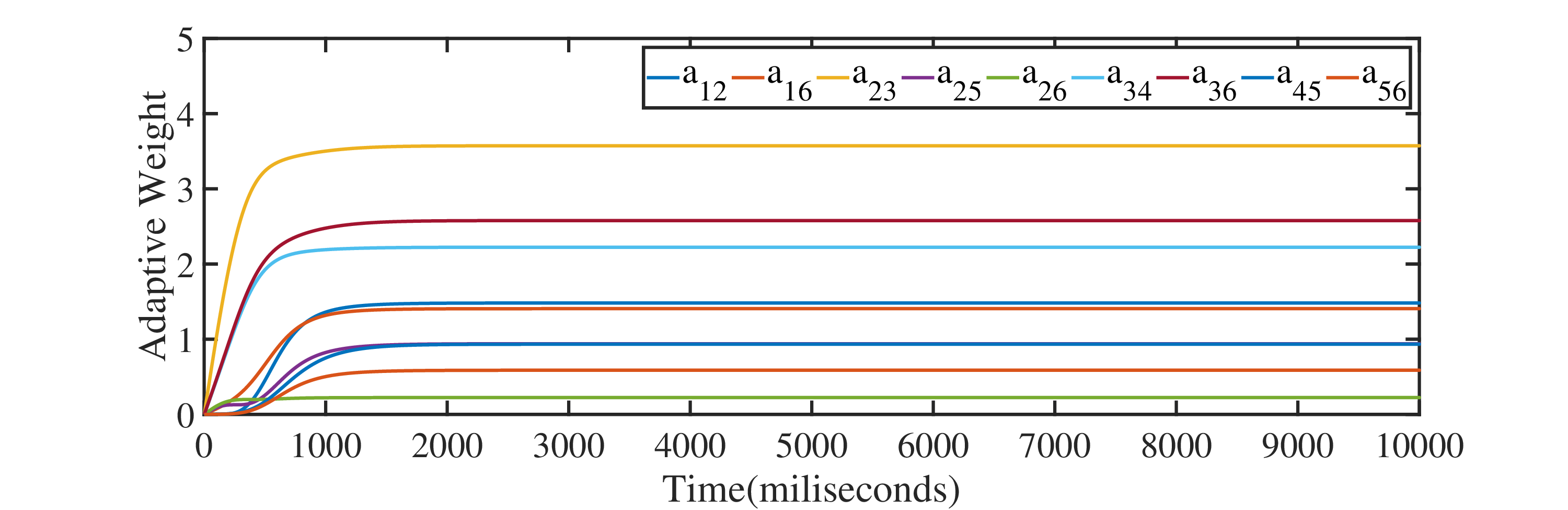}
	\caption{Adaptive weights $a_{ij}, i,j=1,2,...,N$}
\end{figure}

\begin{figure}[!ht]
	\centering
    \captionsetup{justification=centering}
	%\rule{12.8cm}{7.2cm}
	\includegraphics[width=0.99\linewidth,height=40mm, trim=10 0 70 10,clip]{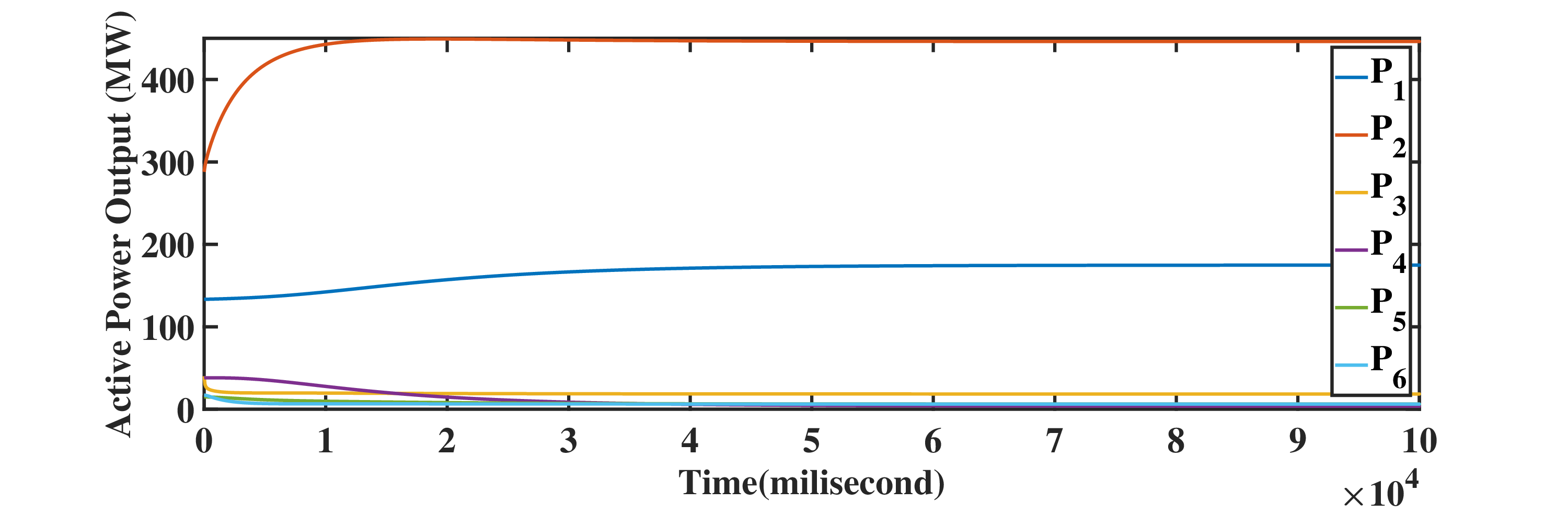}
	\caption{Output powers}
\end{figure}

\subsection{IEEE 30 Bus System with Generator Limits with Fixed Topology}\label{5.2}
 Simulation example is provided in this part to demonstrate the validity of the suggested distributed economic power dispatch algorithms. Here an optimal economic dispatch of IEEE-30 bus system \cite{19} with generating limits is demonstrated using the proposed approach. The parameters of the cost function are taken from Table 1, and the $\delta$ is considered as 10.

 \par The power demand is assumed to be $P_D=513.34 MW$, and the initial generating powers of all units, in this case, are ${P_0} = \left\{ {133.36,287.98,40,20,15,17} \right\}$. The modified incremental cost consensus can be verified in Fig 6. Moreover, in this case from Fig 7, it can be verified that the  adaptive  weights again settle down some positive fixed higher values and  from Fig. 8. it can be verified that all the active power outputs are always within the specified limit.

\begin{figure}[!ht]
	\centering
	%\rule{12.8cm}{7.2cm}
	\includegraphics[width=0.99\linewidth,height=40mm, trim=10 0 70 10,clip]{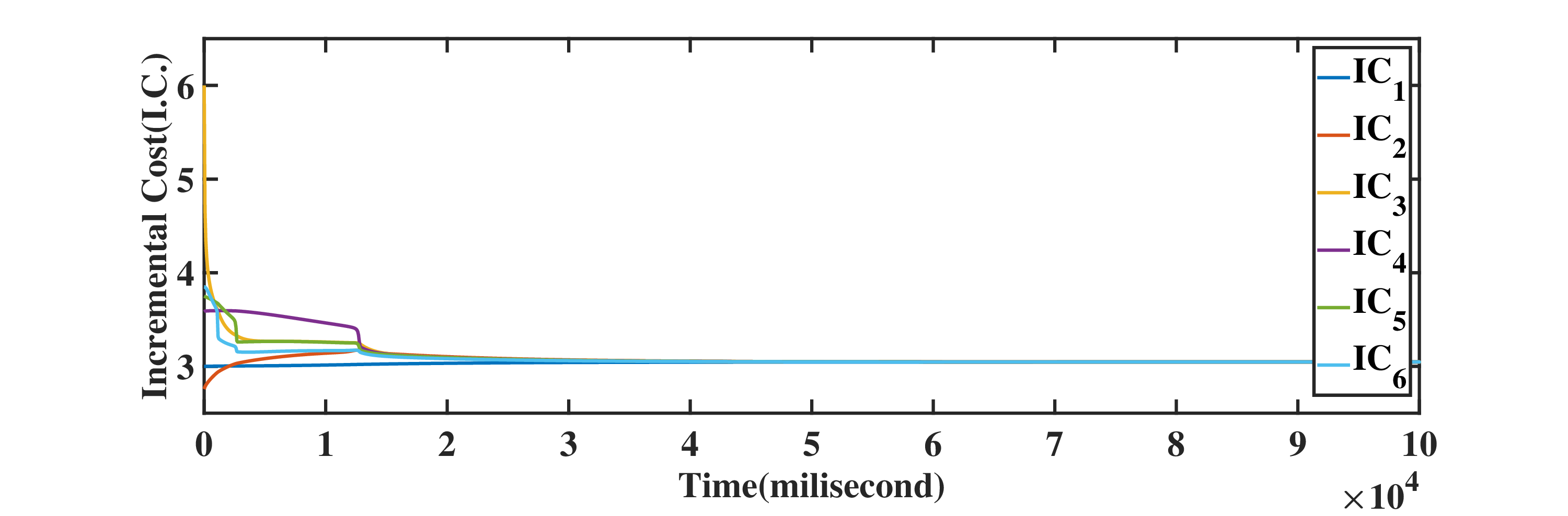}
	\caption{Incremental cost consensus}
\end{figure}

\begin{figure}[!ht]
	\centering
	%\rule{12.8cm}{7.2cm}
	\includegraphics[width=0.99\linewidth,height=40mm, trim=10 0 70 10,clip]{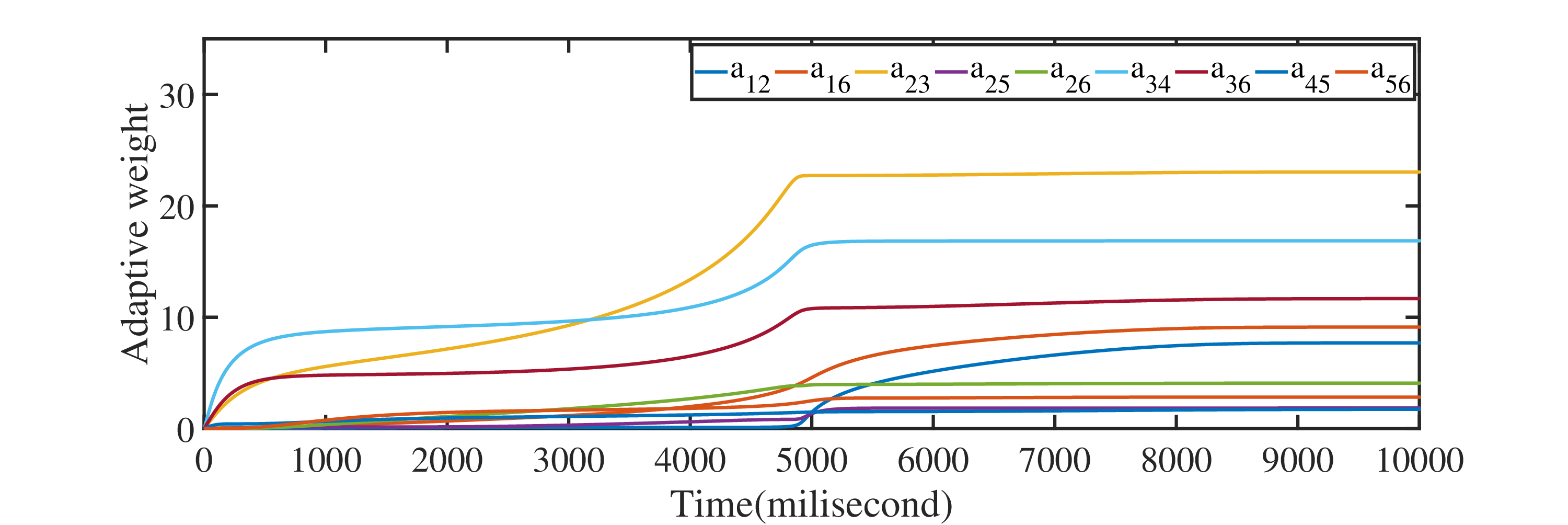}
	\caption{Adaptive weights $a_{ij}, i,j=1,2,...,N$}
\end{figure}

\begin{figure}[!ht]
	\centering
	%\rule{12.8cm}{7.2cm}
	\includegraphics[width=0.99\linewidth,height=40mm, trim=10 0 70 10,clip]{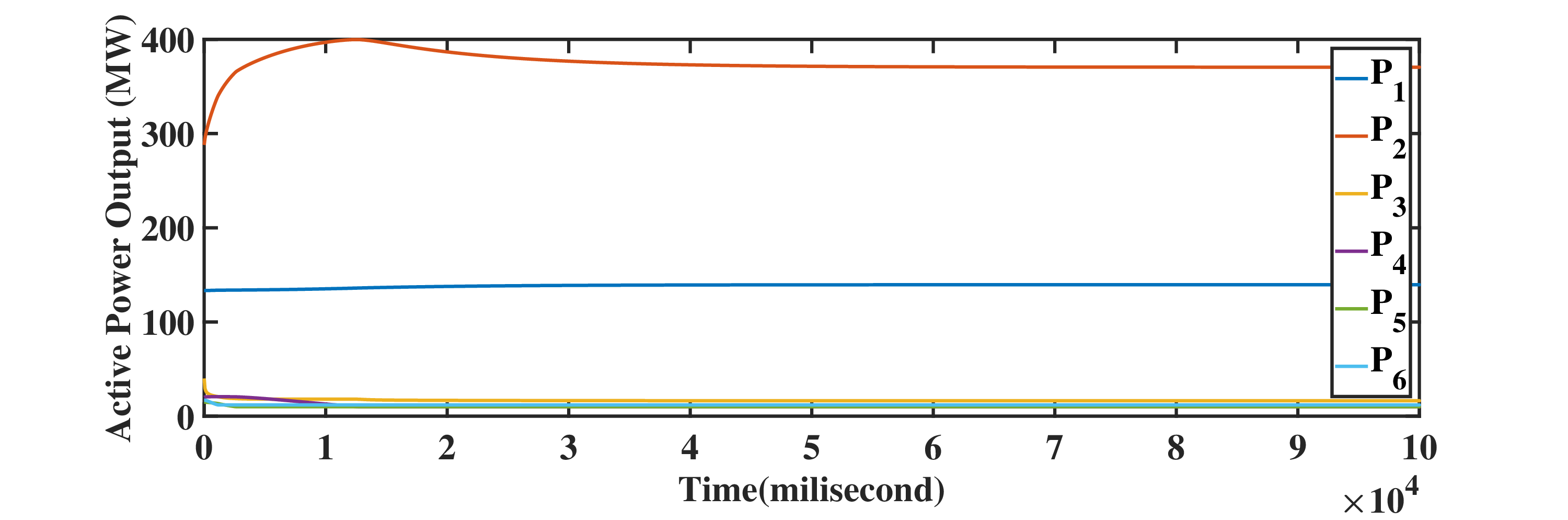}
	\caption{Output powers}
\end{figure}

\subsection{IEEE 30 Bus System: Generator Link Failure}
 Simulation example is provided in this part to demonstrate the validity of the suggested distributed economic power dispatch algorithms considering generator link failure, utilizing switching topologies. Here an optimal economic dispatch of IEEE-30 bus system \cite{19} with generating limits is demonstrated using the proposed approach. The parameters of the cost function are taken from Table 1, and the $\delta$ is assumed to be 10.

 \par The power demand is assumed to be $P_D=513.34 MW$, and the initial generating powers of all units, in this case, are ${P_0} = \left\{ {133.36,287.98.98,40,20,15,17} \right\}$. The topologies switch after every 10 seconds sequentially as can be seen from Fig. 9.  Similar to the previous case, Fig. 10 depicts the incremental cost consensus in this instance. As a result, it demonstrates that even when the network topology alters, the incremental consensus is still reached. Similarly, Fig. 11 illustrates the adaptive weights, and Fig. 12 shows the active power output once more within a given range.

\begin{figure}[!ht]
	\centering
	%\rule{12.8cm}{7.2cm}
	\includegraphics[scale=0.4]{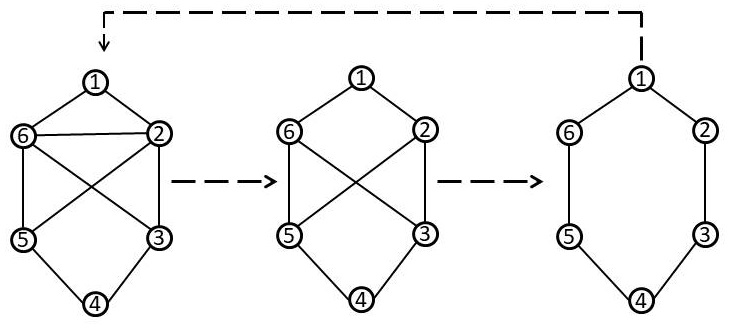}
	\caption{Network Topology}
\end{figure}

\begin{figure}[!ht]
	\centering
	%\rule{12.8cm}{7.2cm}
	\includegraphics[width=0.99\linewidth,height=40mm, trim=10 0 70 10,clip]{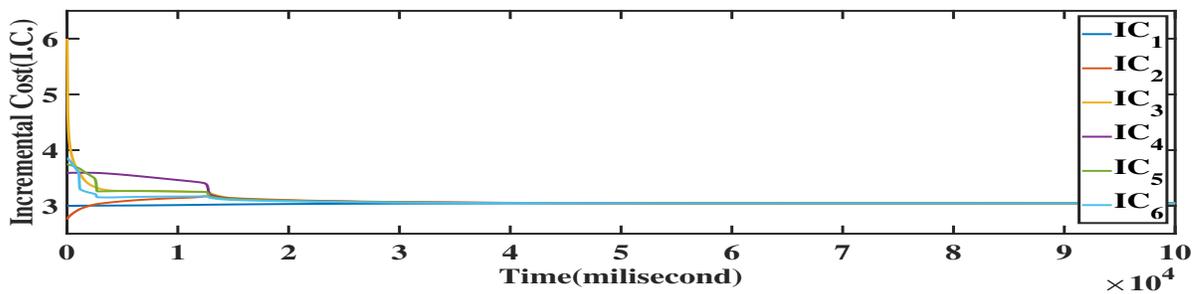}
	\caption{Incremental cost consensus}
\end{figure}

\begin{figure}[!ht]
	\centering
	%\rule{12.8cm}{7.2cm}
	\includegraphics[width=0.99\linewidth,height=40mm, trim=10 0 70 10,clip]{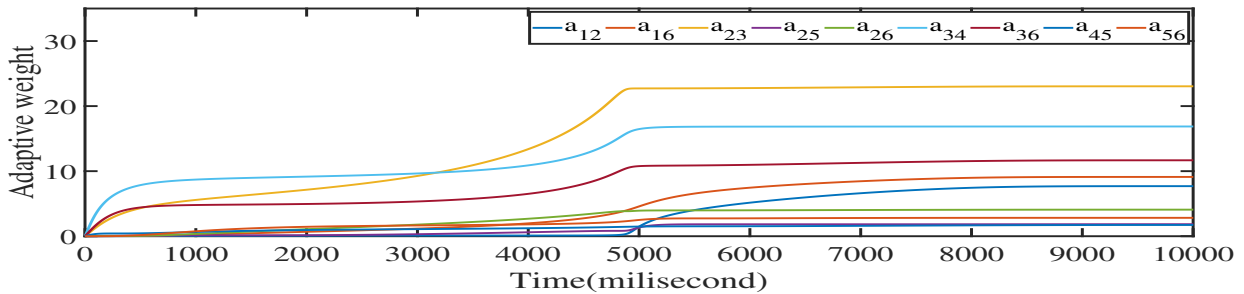}
	\caption{Adaptive weights $a_{ij}, i,j=1,2,...,N$}
\end{figure}

\begin{figure}[!ht]
	\centering
	%\rule{12.8cm}{7.2cm}
	\includegraphics[width=0.99\linewidth,height=40mm, trim=10 0 70 10,clip]{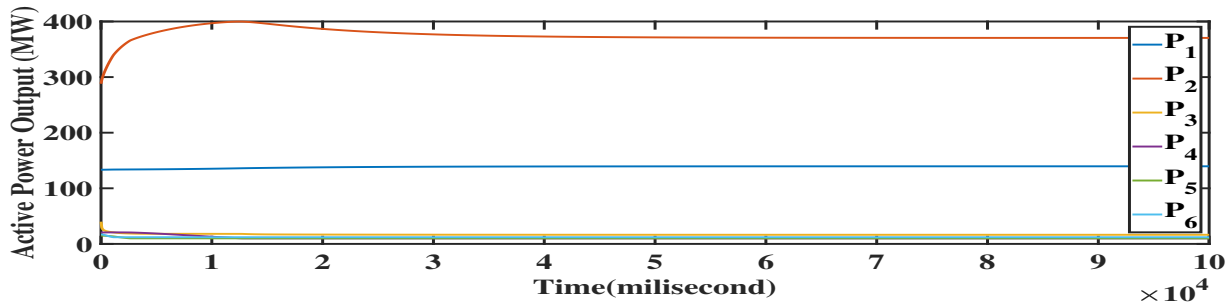}
	\caption{Output powers}
\end{figure}

\subsection{IEEE 30 Bus System with Dummy Nodes: Load Demand Injection in Dummy Nodes }
 To illustrate the concept of Time-Varying Load Demand, a simulation example is presented in this section. To ensure that the drop across the dummy node remains insignificant, the bias value $b_i$ is set to zero, while the value of $c_i$ is kept relatively high, specifically at 0.1.
 \par To cater to an increase in demand a total of 50 MW active power is injected in the dummy nodes at 50 seconds.
 
 The power demand and the initial generating powers of all units are the same as in section \ref{5.2}. Fig. 13 now confirms that incremental cost consensus is reached in the event of a sudden change in load behaviour. Fig. 14 confirms the adaptive coupling weights, and Fig. 15 confirms that the active power outputs are within the designated limits. However, in this instance, the abrupt change in load behaviour causes the incremental cost consensus  and active power output  to change their magnitude at 50 miliseconds.

\begin{figure}[!ht]
	\centering
	%\rule{12.8cm}{7.2cm}
	\includegraphics[width=0.99\linewidth,height=40mm, trim=10 0 70 10,clip]{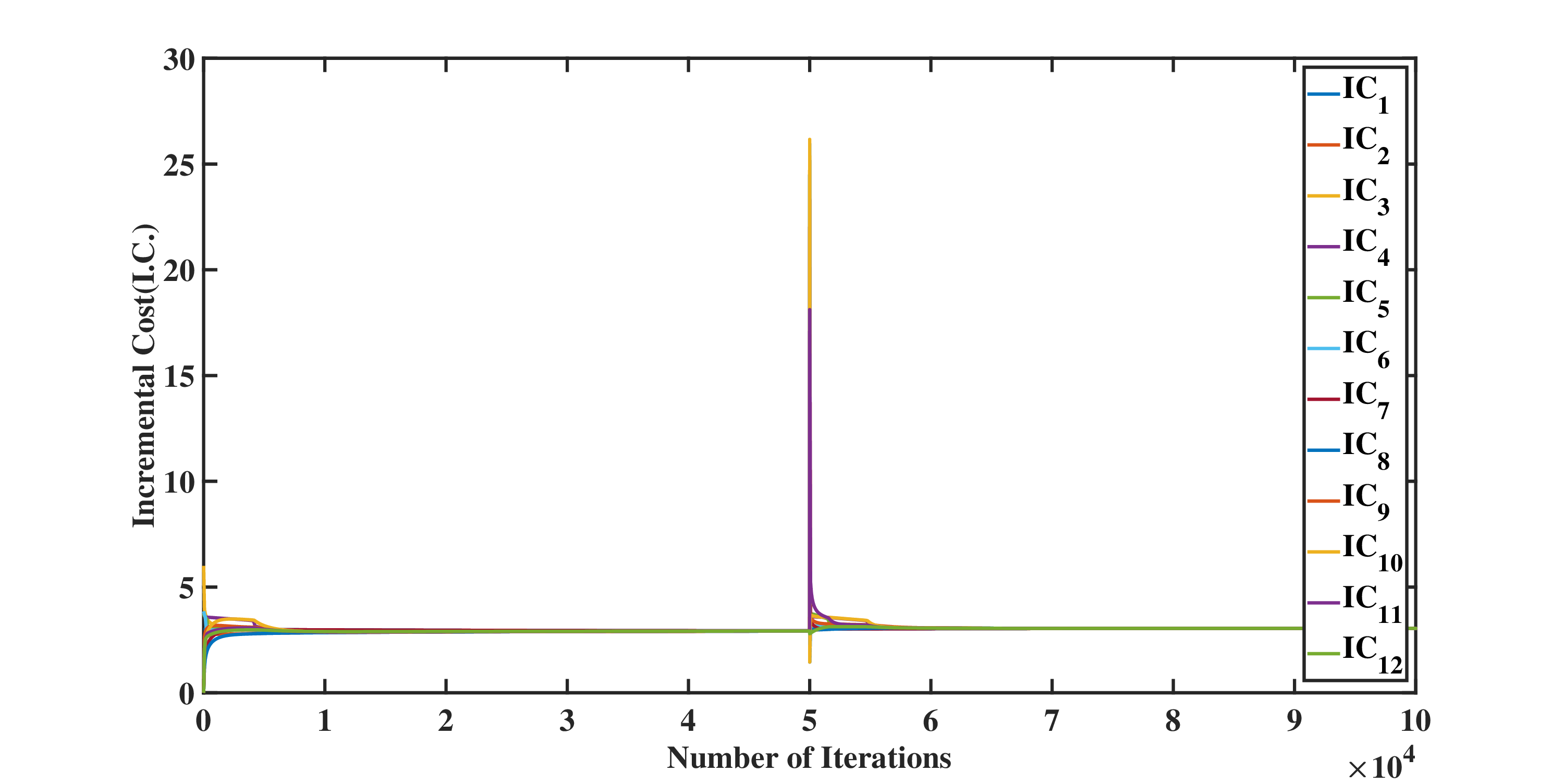}
	\caption{Incremental cost consensus}
\end{figure}

\begin{figure}[!ht]
	\centering
	%\rule{12.8cm}{7.2cm}
	\includegraphics[width=0.99\linewidth,height=40mm, trim=10 0 70 10,clip]{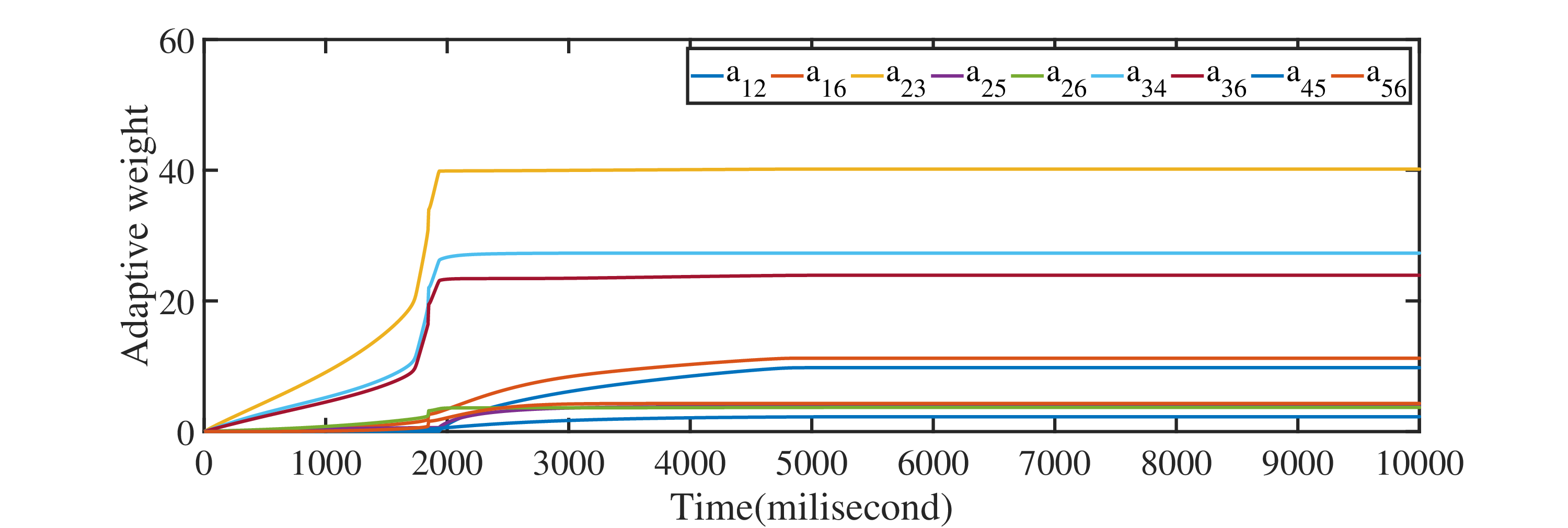}
	\caption{Adaptive weights $a_{ij}, i,j=1,2,...,N$}
\end{figure}

\begin{figure}[!ht]
	\centering
	%\rule{12.8cm}{7.2cm}
	\includegraphics[width=0.99\linewidth,height=40mm, trim=10 0 70 10,clip]{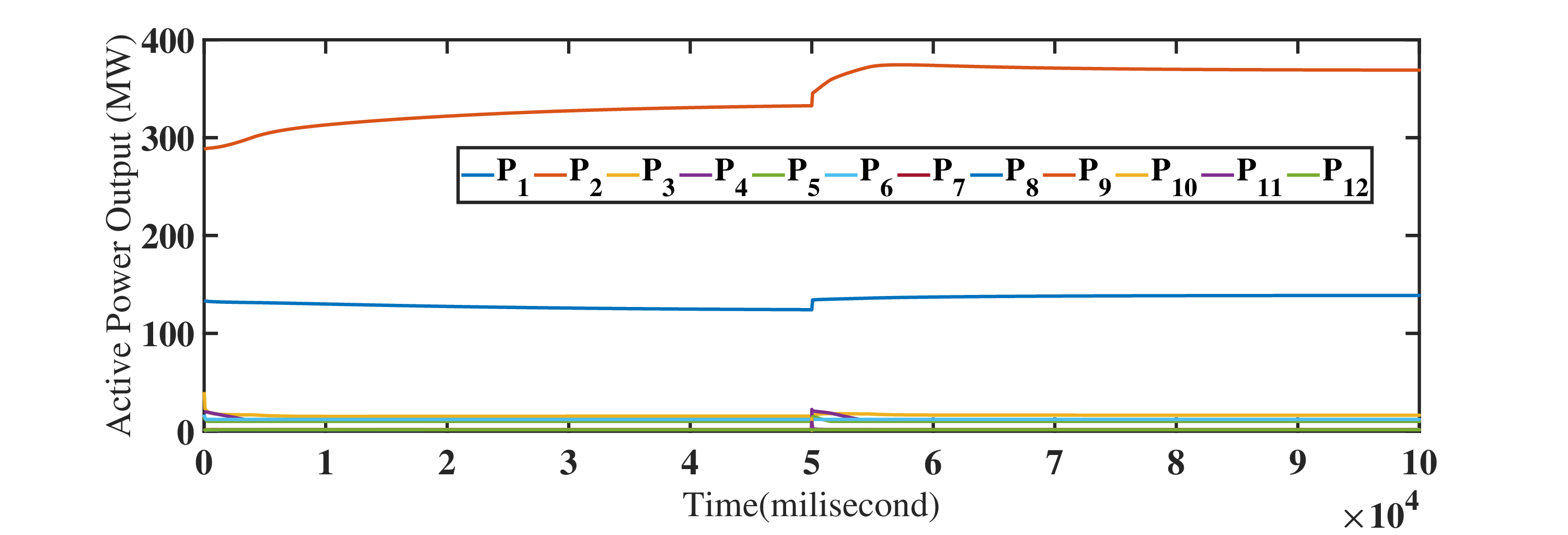}
	\caption{Output powers}
\end{figure}

\section{Conclusion}
This work employs adaptive consensus-based dispatch algorithms to investigate the economic dispatch problem for smart grids with and without generator limits. Distributed consensus based power dispatch algorithms are created to ensure that individual incremental costs are agreed upon while achieving optimal dispatch of active power by distributing the load among the generating units in the proper manner. Additional investigation of time-varying power demand and generator link failure is also done. Our future study will include some realistic models of distributed economic power dispatch techniques in smart grids, such as network constraints i.e. time-delay, power losses in the generator network.

\bibliographystyle{IEEEtran}
\bibliography{ref}
\end{document}